%% file: main.tex
\documentclass[fleqn]{llncs}

\usepackage{verbatim}

\usepackage{amsmath}
\usepackage{amssymb}
\usepackage{stmaryrd} 
\usepackage{wasysym}
\usepackage{mathtools} 
\usepackage{bm}

\usepackage{cite}
\usepackage{hyperref}

\usepackage{tikz}
\usetikzlibrary{arrows,automata,positioning}

\usepackage{graphicx}
\usepackage{subcaption}

\usepackage{ltltikz}

\usepackage{multirow}

\usepackage[firstpage]{draftwatermark}\SetWatermarkText{\hspace*{5.2in}\raisebox{5in}{\includegraphics[scale=0.1]{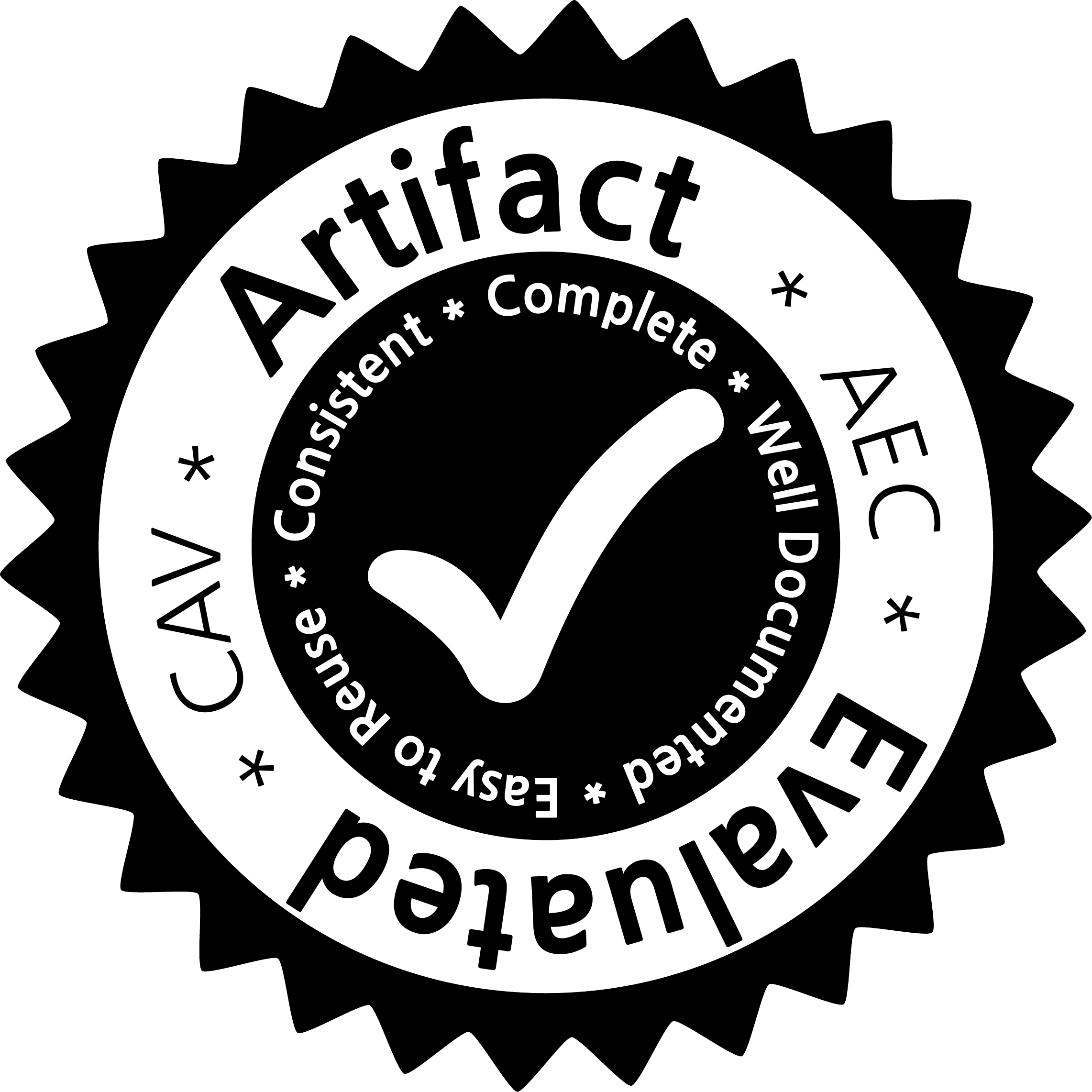}}}\SetWatermarkAngle{0}

\input{./commands}

\title{Synthesizing Reactive Systems from Hyperproperties\thanks{Supported by the European Research Council (ERC) Grant OSARES (No.\ 683300).}}

\author{Bernd Finkbeiner \and Christopher Hahn \and Philip Lukert \and Marvin Stenger \and Leander Tentrup}

\institute{Reactive Systems Group\\Saarland University\\\email{lastname@react.uni-saarland.de}}

\begin{document}

\maketitle

\input{./abstract}

\input{./introduction}

\input{./preliminaries}

\input{./synthesis}

\input{./bounded}
\input{./unrealizability}

\input{./experiments}

\input{./conclusion}

\bibliographystyle{splncs03}
\bibliography{main}


\end{document}

%% file: commands.tex
\newcommand{\set}[1]{\{#1\}}

\newcommand{\tuple}[1]{{\langle #1 \rangle}}
\newcommand{\card}[1]{{|#1|}}
\newcommand{\ldot}{\mathpunct{.}}

\newcommand{\pow}[1]{2^{#1}}
\newcommand{\<}{\subseteq}
\newcommand{\cupdot}{\mathbin{\dot\cup}}

\newcommand{\pspace}{\textsc{PSpace}}

\newcommand{\threeexptime}{\textsc{3ExpTime}}

\newcommand{\ltl}{\text{LTL}}
\newcommand{\hyperltl}{\text{HyperLTL}}
\newcommand{\secltl}{\text{SecLTL}}

\newcommand{\bool}{\mathbb{B}}
\newcommand{\nat}{\mathbb{N}}

\newcommand{\lang}{\mathcal{L}}
\newcommand{\ap}{\text{AP}}
\renewcommand{\models}{\vDash}
\newcommand{\nmodels}{\nvDash}

\newcommand{\cexpaths}{\mathcal{P}}

\newcommand{\ucw}{\mathcal{A}}
\newcommand{\zip}{\mathit{zip}}
\newcommand{\unzip}{\mathit{unzip}}

\newcommand{\U}{\Until}
\newcommand{\X}{\Next}
\newcommand{\G}{\Globally}
\newcommand{\F}{\Eventually}
\newcommand{\R}{\Release}
\newcommand{\W}{\WUntil}
\newcommand{\true}{\mathit{true}}

\newcommand{\btrue}{\top}

\newcommand{\pathassign}{\Pi}
\newcommand{\pathvars}{\mathcal{V}}

\newcommand{\dep}[2]{D_{#1 \mapsto #2}}
\newcommand{\collapse}{\mathit{collapse}}
\newcommand{\prj}{\mathit{prj}}

\newcommand{\strat}[2]{(2^{#1})^* \rightarrow 2^{#2}}
\newcommand{\fun}[2]{#1 \rightarrow #2}

\newcommand{\tsys}{\mathcal{S}} 

\newcommand{\arch}{A}
\newcommand{\penv}{{p_\textit{env}}}

\newcommand{\tool}{\text{BoSyHyper}}



%% file: abstract.tex
\begin{abstract}
  We study the reactive synthesis problem for hyperproperties given as formulas of the temporal logic HyperLTL. Hyperproperties generalize trace properties, i.e., sets of traces, to \emph{sets of sets} of traces. Typical examples are information-flow policies like noninterference, which stipulate that no sensitive data must leak into the public domain. Such properties cannot be expressed in standard linear or branching-time temporal logics like LTL, CTL, or CTL$^*$.
  We show that, while the synthesis problem is undecidable for full HyperLTL, it remains decidable for the $\exists^*$, $\exists^*\forall^1$, and the
$\mathit{linear}\;\forall^*$ fragments. Beyond these fragments, the synthesis problem immediately becomes undecidable. For universal HyperLTL, we present a semi-decision procedure that constructs implementations and counterexamples up to a given bound. 
  We report encouraging experimental results obtained with a prototype implementation on example specifications with hyperproperties like symmetric responses, secrecy, and information-flow.
\end{abstract}

%% file: introduction.tex
\section{Introduction}


\emph{Hyperproperties}~\cite{journals/jcs/ClarksonS10} generalize
trace properties in that they not only check the correctness of
\emph{individual} computation traces in isolation, but relate
\emph{multiple} computation traces to each other.
$\hyperltl$~\cite{conf/post/ClarksonFKMRS14} is a logic for expressing
temporal hyperproperties, by extending linear-time temporal logic
(LTL) with \emph{explicit} quantification over traces.  $\hyperltl$
has been used to specify a variety of information-flow and security
properties.  Examples include classical properties like
non-interference and observational determinism, as well as
quantitative information-flow properties, symmetries in hardware
designs, and formally verified error correcting
codes~\cite{conf/cav/FinkbeinerRS15}.  For example, observational
determinism can be expressed as the $\hyperltl$
formula
$
\forall \pi \forall \pi'\ldot \G ( I_\pi = I_{\pi'} ) \rightarrow \G
( O_\pi = O_{\pi'} ),  
$
stating that, for every pair of traces, if the
observable inputs are the same, then the observable outputs must be
same as well.  While the
satisfiability~\cite{conf/concur/FinkbeinerH16}, model checking~\cite{conf/post/ClarksonFKMRS14,conf/cav/FinkbeinerRS15}, and
runtime verification~\cite{conf/csfw/AgrawalB16,conf/rv/FinkbeinerHST17} problem for $\hyperltl$ has been studied, the
\emph{reactive synthesis} problem of $\hyperltl$ is, so far, still
open.

In reactive synthesis, we automatically construct an implementation
that is guaranteed to satisfy a given specification. A fundamental
difference to verification is that there is no human programmer
involved: in verification, the programmer would first produce an
implementation, which is then verified against the specification. In
synthesis, the implementation is directly constructed from the
specification. Because there is no programmer, it is crucial that the
specification contains \emph{all} desired properties of the
implementation: the synthesized implementation is guaranteed to
satisfy the given specification, but nothing is guaranteed beyond
that. The added expressive power of HyperLTL over LTL is very
attractive for synthesis: with synthesis from hyperproperties, we can
guarantee that the implementation does not only accomplish the desired
functionality, but is also free of information leaks, is symmetric, is
fault-tolerant with respect to transmission errors, etc.

More formally, the reactive synthesis problem asks for a
\emph{strategy}, that is a tree branching on environment inputs whose
nodes are labeled by the system output. Collecting the inputs and
outputs along a branch of the tree, we obtain a trace.  If the set of
traces collected from the branches of the strategy tree satisfies the
specification, we say that the strategy \emph{realizes} the
specification.  The specification is \emph{realizable} iff there
exists a strategy tree that realizes the specification.
With LTL specifications, we get trees where the trace on each
individual branch satisfies the LTL formula. With HyperLTL, we
additionally get trees where the traces between different branches are
in a specified relationship.  This is dramatically more
powerful.

Consider, for example, the well-studied \emph{distributed} version of
the reactive synthesis problem, where the system is split into a set
of processes, that each only see a subset of the inputs. The
distributed synthesis problem for LTL can be expressed as the standard
(non-distributed) synthesis problem for HyperLTL, by adding for each
process the requirement that the process output is
\emph{observationally deterministic} in the process input.  HyperLTL
synthesis thus subsumes distributed synthesis.  The information-flow
requirements realized by HyperLTL synthesis can, however, be much more
sophisticated than the observational determinism needed for
distributed synthesis. Consider, for example, the \emph{dining
  cryptographers} problem~\cite{journals/cacm/Chaum85}: three cryptographers $C_a,C_b,$ and
$C_c$ sit at a table in a restaurant having dinner and either one of
cryptographers or, alternatively, the NSA must pay for their meal.  Is
there a protocol where each cryptographer can find out whether it was
a cryptographer who paid or the NSA, but cannot find out which
cryptographer paid the bill?

Synthesis from LTL formulas is known to be decidable in doubly
exponential time.  The fact that the distributed synthesis problem is
undecidable~\cite{conf/focs/PnueliR90} immediately eliminates the hope for a similar general
result for HyperLTL. However, since LTL is obviously a fragment of
HyperLTL, this immediately leads to the question whether the synthesis
problem is still decidable for fragments of
HyperLTL that are close to LTL but go beyond LTL: when exactly does the synthesis
problem become undecidable? From a more practical point of view,
the interesting question is whether semi-algorithms for distributed
synthesis~\cite{journals/sttt/FinkbeinerS13,conf/tacas/FaymonvilleFRT17}, which have been successful in constructing distributed
systems from LTL specifications despite the undecidability of the general problem,
can be extended to HyperLTL?

In this paper, we answer the first question by studying the
$\exists^*$, $\exists^*\forall^1$, and the
$\mathit{linear}\;\forall^*$ fragment. We show that the synthesis
problem for all three fragments is decidable, and the problem becomes
undecidable as soon as we go beyond these fragments. In particular,
the synthesis problem for the full $\forall^*$ fragment, which includes
observational determinism, is undecidable.

We answer the second question by studying the \emph{bounded} version
of the synthesis problem for the $\forall^*$ fragment. In order to detect realizability, we ask
whether, for a universal HyperLTL formula $\varphi$ and a given bound
$n$ on the number of states, there exists a representation of the
strategy tree as a finite-state machine with no more than $n$ states
that satisfies $\varphi$. To detect unrealizability, we
check whether there exists a counterexample to realizability
of bounded size. 
We show that both checks can be effectively reduced to
SMT solving.

\subsubsection{Related Work.}

  


$\hyperltl$~\cite{conf/post/ClarksonFKMRS14} is a successor of the temporal logic $\secltl$~\cite{conf/vmcai/DimitrovaFKRS12} used to characterize temporal information-flow.
The model-checking~\cite{conf/post/ClarksonFKMRS14,conf/cav/FinkbeinerRS15}, satisfiability~\cite{conf/concur/FinkbeinerH16}, monitoring problem~\cite{conf/csfw/AgrawalB16,conf/rv/FinkbeinerHST17}, and the first-order extension~\cite{conf/stacs/Finkbeiner017} of $\hyperltl$ has been studied before.
To the best of the authors knowledge, this is the first work that considers the synthesis problem for temporal hyperproperties.
We base our algorithms on well-known synthesis algorithms such as bounded synthesis~\cite{journals/sttt/FinkbeinerS13} that itself is an instance of Safraless synthesis~\cite{conf/focs/KupfermanV05} for $\omega$-regular languages.
Further techniques that we adapt for hyperproperties are lazy synthesis~\cite{conf/vmcai/FinkbeinerJ12} and the bounded unrealizability method~\cite{conf/tacas/FinkbeinerT14,journals/corr/FinkbeinerT15}.

Hyperproperties~\cite{journals/jcs/ClarksonS10} can be seen as a unifying framework for many different properties of interest in multiple distinct areas of research.
Information-flow properties in security and privacy research are hyperproperties~\cite{conf/post/ClarksonFKMRS14}.
$\hyperltl$ subsumes logics that reason over knowledge~\cite{conf/post/ClarksonFKMRS14}.
Information-flow in distributed systems is another example of hyperproperties, and the $\hyperltl$ realizability problem subsumes both the distributed synthesis problem~\cite{conf/focs/PnueliR90,conf/lics/FinkbeinerS05} as well as synthesis of fault-tolerant systems~\cite{journals/corr/FinkbeinerT15}.
In circuit verification, the semantic independence of circuit output signals on a certain set of inputs, enabling a range of potential optimizations, is a hyperproperty.


%% file: preliminaries.tex
\section{Preliminaries}

\paragraph{HyperLTL.}

$\hyperltl$~\cite{conf/post/ClarksonFKMRS14} is a temporal logic for specifying hyperproperties.
It extends $\ltl$ by quantification over trace variables $\pi$ and a method to link atomic propositions to specific traces.
The set of trace variables is $\pathvars$.
Formulas in $\hyperltl$ are given by the grammar
\begin{align*}
\varphi &{}\Coloneqq \forall\pi\ldot\varphi \mid \exists\pi\ldot\varphi \mid \psi \enspace, \text{ and}\\
\psi &{}\Coloneqq a_\pi \mid \neg\psi \mid \psi\lor\psi \mid \X\psi \mid \psi\U\psi \enspace,
\end{align*}
where $a \in \ap$ and $\pi \in \pathvars$.
The alphabet of a $\hyperltl$ formula is $2^\mathit{AP}$.
We allow the standard boolean connectives $\wedge$, $\rightarrow$, $\leftrightarrow$ as well as the derived $\ltl$ operators release $\varphi \R \psi \equiv \neg(\neg\varphi \U \neg\psi)$, eventually $\F \varphi \equiv \true \U \varphi$, globally $\G \varphi \equiv \neg \F \neg \varphi$, and weak until $\varphi \W \psi \equiv \G \varphi \lor (\varphi \U \psi)$.

The semantics is given by the satisfaction relation $\models_T$ over a set of traces $T \subseteq (2^\ap)^\omega$.
We define an assignment $\pathassign : \pathvars \to (2^\ap)^\omega$ that maps trace variables to traces. $\pathassign[i,\infty]$ is the trace assignment that is equal to $\pathassign(\pi)[i,\infty]$ for all $\pi$ and denotes the assignment where the first $i$ items are removed from each trace.

$
\begin{array}{ll}
  \pathassign \models_T a_\pi         \qquad \qquad & \text{if } a \in \pathassign(\pi)[0] \\
  \pathassign \models_T \neg \varphi              & \text{if } \pathassign \nmodels_T \varphi \\
  \pathassign \models_T \varphi \lor \psi         & \text{if } \pathassign \models_T \varphi \text{ or } \pathassign \models_T \psi \\
  \pathassign \models_T \X \varphi                & \text{if } \pathassign[1,\infty] \models_T \varphi \\
  \pathassign \models_T \varphi\U\psi             & \text{if } \exists i \geq 0 \ldot \pathassign[i,\infty] \models_T \psi \land \forall 0 \leq j < i \ldot \pathassign[j,\infty] \models_T \varphi \\
  \pathassign \models_T \exists \pi \ldot \varphi & \text{if there is some } t \in T \text{ such that } \pathassign[\pi \mapsto t] \models_T \varphi\\
  \pathassign \models_T \forall \pi \ldot \varphi & \text{if for all } t \in T \text{ holds that } \pathassign[\pi \mapsto t] \models_T \varphi
\end{array}$\\[\bigskipamount]
We write $T \models \varphi$ for $\set{} \models_T \varphi$ where $\set{}$ denotes the empty assignment.
Two $\hyperltl$ formulas $\varphi$ and $\psi$ are equivalent, written $\varphi \equiv \psi$ if they have the same models.

\emph{(In)dependence} is a common hyperproperty for which we define the following syntactic sugar.
Given two disjoint subsets of atomic propositions $C \subseteq \ap$ and $A \subseteq \ap$, we define independence as the following $\hyperltl$ formula
\begin{equation}
  \dep{A}{C}\coloneqq
  \forall \pi \forall \pi' \ldot
  \left(
  \bigvee_{a \in A} (a_\pi \nleftrightarrow a_{\pi'})
  \right)
  \R
  \left(
  \bigwedge_{c \in C} ( c_\pi \leftrightarrow c_{\pi'})
  \right)
  \enspace.
\end{equation}
This guarantees that every proposition $c \in C$ solely depends on propositions $A$.

\paragraph{Strategies.}
A \emph{strategy} $f \colon \strat{I}{O}$ maps sequences of input valuations $\pow{I}$ to an output valuation $\pow{O}$.
The behavior of a strategy $f\colon \strat{I}{O}$ is characterized by an infinite tree that branches by the valuations of $I$ and whose nodes $w \in (\pow{I})^*$ are labeled with the strategic choice $f(w)$.
For an infinite word $w = w_0 w_1 w_2 \cdots \in (\pow{I})^\omega$, the corresponding labeled path is defined as $(f(\epsilon) \cup w_0)(f(w_0) \cup w_1)(f(w_0 w_1) \cup w_2)\cdots \in (2^{I \cup O})^\omega$.
We lift the set containment operator $\in$ to the containment of a labeled path $w = w_0 w_1 w_2 \cdots \in (2^{I \cup O})^\omega$ in a strategy tree induced by $f \colon \strat{I}{O}$, i.e., $w \in f$ if, and only if, $f(\epsilon) = w_0 \cap O$ and $f((w_0 \cap I) \cdots (w_i \cap I)) = w_{i+1} \cap O$ for all $i \geq 0$.
We define the satisfaction of a $\hyperltl$ formula $\varphi$ (over propositions $I \cup O$) on strategy~$f$, written $f \models \varphi$, as $\set{w \mid w \in f} \models \varphi$.
Thus, a strategy $f$ is a model of $\varphi$ if the set of labeled paths of $f$ is a model of $\varphi$.

%% file: synthesis.tex
\section{$\hyperltl$ Synthesis} \label{sec:decidability}
In this section, we identify fragments of $\hyperltl$ for which the realizability problem is decidable.
Our findings are summarized in Table~\ref{table:decidability}.

\begin{table}[t]
  \caption{Summary of decidability results.}
  \label{table:decidability}
  \centering
  \begin{tabular}{cccccc}
    \hline\noalign{\smallskip}
     $\exists^*$ & $\exists^* \forall^1$ & \hspace{10pt} $\exists^* \forall^{>1}$ \hspace{10pt} & \hspace{10pt}$\forall^*$\hspace{10pt} & \hspace{10pt}$\forall^* \exists^*$\hspace{10pt} & $\mathit{linear}\;\forall^*$ \\ \noalign{\smallskip} \hline \noalign{\smallskip}
    $\pspace$-complete \hspace{10pt} & $\threeexptime$ & \multicolumn{3}{c}{undecidable} & decidable \\ \noalign{\smallskip} \hline
  \end{tabular}
\end{table}

\begin{definition}[$\hyperltl$ Realizability]
  A $\hyperltl$ formula $\varphi$ over atomic propositions $\ap = I \cupdot O$ is realizable if there is a strategy $f \colon \strat{I}{O}$ that satisfies $\varphi$.
\end{definition}
We base our investigation on the structure of the quantifier prefix of the HyperLTL formulas. We call a HyperLTL formula $\varphi$ (quantifier) \emph{alternation-free} if the quantifier prefix consists solely of either universal or existential quantifiers. We denote the corresponding fragments as the (universal) $\forall^*$ and the (existential) $\exists^*$ fragment, respectively.
A HyperLTL formula is in the $\exists^* \forall^*$ fragment, if it starts with arbitrarily many existential quantifiers, followed by arbitrarily many universal quantifiers.
Respectively for the $\forall^*\exists^*$ fragment.
For a given natural number $n$, we refer to a bounded number of quantifiers with $\forall^n$, respectively $\exists^n$.
The $\forall^1$ realizability problem is equivalent to the $\ltl$ realizability problem.

\subsubsection{$\exists^*$ Fragment.}

We show that the realizability problem for existential $\hyperltl$ is $\pspace$-complete.
We reduce the realizability problem to the satisfiability problem for bounded one-alternating $\exists^*\forall^2 \hyperltl$~\cite{conf/concur/FinkbeinerH16}, i.e., finding a trace set $T$ such that $T \models \varphi$.

\begin{lemma}\label{lemma:EdetSynthesis=EAsat} 
	An existential $\hyperltl$ formula $\varphi$ is realizable if, and only if,
	$\psi \coloneqq \varphi \land \dep{I}{O}$
	is satisfiable.
\end{lemma}

\begin{proof}
	Assume $f \colon \strat{I}{O}$ realizes $\varphi$, that is $f \models \varphi$.
	Let $T = \set{w \mid w \in f}$ be the set of traces generated by $f$.
	It holds that $T \models \varphi$ and $T \models \dep{I}{O}$.
	Therefore, $\psi$ is satisfiable.
	
	Assume $\psi$ is satisfiable.
	Let $S$ be a set of traces that satisfies $\psi$.
	We construct a strategy $f \colon \strat{I}{O}$ as
	\begin{equation*}
      f(\sigma) =
      \begin{cases}
	  w_\card{\sigma} \cap O & \text{if } \sigma \text{ is a prefix of some } w|_I \text{ with } w \in S \enspace \text{, and} \\
	  \emptyset & \text{otherwise} \enspace.
	  \end{cases}
	\end{equation*}
	Where $w|_I$ denotes the trace restricted to $I$, formally $w_i \cap I$ for all $i \geq 0$. 
	Note that if there are multiple candidates $w \in S$, then $w_\card{\sigma} \cap O$ is the same for all of them because of the required non-determinism $\dep{I}{O}$.
	By construction, all traces in $S$ are contained in $f$ and with $S \models \varphi$ it holds that $f \models \varphi$ as $\varphi$ is an existential formula.
\end{proof}

\begin{theorem}
	Realizability of existential $\hyperltl$ specifications is decidable.
\end{theorem}

\begin{proof}
	The formula $\psi$ from Lemma~\ref{lemma:EdetSynthesis=EAsat} is in the $\exists^*\forall^2$ fragment, for which satisfiability is decidable~\cite{conf/concur/FinkbeinerH16}.
\end{proof}

\begin{corollary}
	Realizability of $\exists^* \hyperltl$ specifications is $\pspace$-complete.
\end{corollary}

\begin{proof}
	Given an existential $\hyperltl$ formula, we gave a linear reduction to the satisfiability of the $\exists^*\forall^2$ fragment in Lemma~\ref{lemma:EdetSynthesis=EAsat}.
	The satisfiability problem for a bounded number of universal quantifiers is in $\pspace$~\cite{conf/concur/FinkbeinerH16}.
	Hardness follows from $\ltl$ satisfiability, which is equivalent to the $\exists^1$ fragment.
\end{proof}

\subsubsection{$\forall^*$ Fragment.}

In the following, we will use the \emph{distributed synthesis} problem, i.e., the problem whether there is an implementation of processes in a distributed \emph{architecture} that satisfies an $\ltl$ formula.
Formally, a distributed architecture $\arch$ is a tuple $\tuple{P, p_{env}, \mathcal{I}, \mathcal{O}}$ where
$P$ is a finite set of processes with distinguished environment process $p_\textit{env} \in P$. The functions $\mathcal{I} \colon \fun{P}{\pow{\ap}}$ and $\mathcal{O} \colon \fun{P}{\pow{\ap}}$ define the inputs and outputs of processes.
While processes may share the same inputs (in case of broadcasting), the outputs of processes must be pairwise disjoint, i.e., for all $p \neq p' \in P$ it holds that $\mathcal{O}(p) \cap \mathcal{O}(p') = \emptyset$.
W.l.o.g.~we assume that $\mathcal{I}(\penv) = \emptyset$.
The distributed synthesis problem for architectures without \emph{information forks}~\cite{conf/lics/FinkbeinerS05} is decidable.
Example architectures are depicted in Fig.~\ref{fig:distributed-architectures}.
The architecture in Fig.~\ref{fig:informationFork} contains an information fork while the architecture in Fig.~\ref{fig:architecture-incomplete-information} does not.
Furthermore, the processes in Fig.~\ref{fig:architecture-incomplete-information} can be ordered linearly according to the subset relation on the inputs.

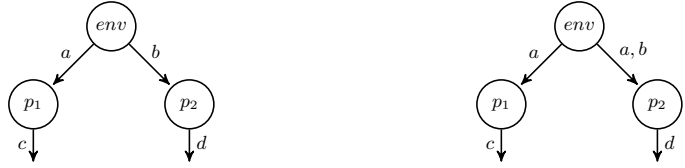
\begin{figure}[t]
	\begin{subfigure}[t]{0.49\textwidth}
		\centering
		\input{tikz/undecidable_architecture}
		\caption{An architecture of two processes that specify process $p_1$ to produce $c$ from $a$ and $p_2$ to produce $d$ from $b$.} 
		\label{fig:informationFork}
	\end{subfigure}\hfill
	\begin{subfigure}[t]{0.49\textwidth}
		\centering
		\input{tikz/incomplete_information}
		\caption{The same architecture as on the left, where only the inputs of process $p_2$ are changed to $a$ and $b$.}
		\label{fig:architecture-incomplete-information}
	\end{subfigure}
	\caption{Distributed architectures}
	\label{fig:distributed-architectures}
\end{figure}

\begin{theorem}
	The synthesis problem for universal $\hyperltl$ is undecidable.
\end{theorem}

\begin{proof}
	In the $\forall^*$ fragment (and thus in the $\exists^*\forall^*$ fragment), we can encode a distributed architecture~\cite{conf/lics/FinkbeinerS05}, for which $\ltl$ synthesis is undecidable. In particular, we can encode the architecture shown in Fig.~\ref{fig:informationFork}. This architecture basically specifies $c$ to depend only on $a$ and analogously $d$ on $b$. That can be encoded by $\dep{\{a\}}{\{c\}}$ and $\dep{\{b\}}{\{d\}}$. The $\ltl$ synthesis problem for this architecture is already shown to be undecidable~\cite{conf/lics/FinkbeinerS05}, i.e., given an $\ltl$ formula over $I=\{a,b\}$ and $O=\{c,d\}$, we cannot automatically construct processes $p_1$ and $p_2$ that realize the formula.
\end{proof}

\subsubsection{Linear $\forall^*$ Fragment.}

For characterizing the linear fragment of HyperLTL, we will present a transformation from a formula with arbitrarily many universal quantifiers to a formula with only one quantifier.
This transformation collapses the universal quantifier into a single one and renames the path variables accordingly.
For example, $\forall \pi_1 \forall \pi_2 \ldot \G a_{\pi_1} \lor \G a_{\pi_2}$ is transformed into an equivalent $\forall^1$ formula $\forall \pi \ldot \G a_\pi \lor \G a_\pi$.
However, this transformation does not always produce equivalent formulas as $\forall \pi_1 \forall \pi_2 \ldot \G (a_{\pi_1} \leftrightarrow a_{\pi_2})$ is not equivalent to its collapsed form $\forall \pi\ldot \G (a_\pi \leftrightarrow a_\pi)$.

Let $\varphi$ be $ \forall \pi_1 \cdots \forall \pi_n \ldot \psi$.
We define the collapsed formula of $\varphi$ as $\collapse(\varphi) \coloneqq \forall \pi \ldot \psi[\pi_1 \mapsto \pi][\pi_2 \mapsto \pi]\dots[\pi_n \mapsto \pi]$ where $\psi[\pi_i \mapsto \pi]$ replaces all occurrences of $\pi_i$ in $\psi$ with $\pi$.
Although the collapsed term is not always equivalent to the original formula, we can use it as an indicator whether it is possible at all to express a universal formula with only one quantifier as stated in the following lemma.

\begin{lemma}
	Either $\varphi \equiv \collapse(\varphi)$ or $\varphi$ has no equivalent $\forall^1$ formula.
\end{lemma}

\begin{proof}
	Suppose there is some $\psi \in \forall^1$ with $\psi \equiv \varphi$.
	We show that $\psi \equiv \collapse(\varphi)$.
	Let $T$ be an arbitrary set of traces.
	Let $\mathcal{T} = \set{\set{w} \mid w \in T }$.
	Because $\psi \in \forall^1$, $T \models \psi$ is equivalent to $\forall T' \in \mathcal{T} \ldot T' \models \psi$, which is by assumption equivalent to $\forall T' \in \mathcal{T} \ldot T' \models \varphi$.
	Now, $\varphi$ operates on singleton trace sets only.
	This means that all quantified paths have to be the same, which yields that we can use the same path variable for all of them.
	So $\forall T' \in \mathcal{T} \ldot T' \models \varphi \leftrightarrow  T' \models \collapse(\varphi)$ that is again equivalent to $T \models \collapse(\varphi)$.
	Because $\psi \equiv \collapse(\varphi)$ and $\psi \equiv \varphi$ it holds that $\varphi \equiv \collapse(\varphi)$.
\end{proof}

The $\ltl$ realizability problem for distributed architectures without information forks~\cite{conf/lics/FinkbeinerS05} are decidable. These architectures are in some way \emph{linear}, i.e., the processes can be ordered such that lower processes always have a subset of the information of upper processes.
The linear fragment of universal $\hyperltl$ addresses exactly these architectures.

In the following, we sketch the characterization of the linear fragment of $\hyperltl$.
Given a formula $\varphi$, we seek for variable dependencies of the form $\dep{J}{\set{o}}$ with $J\<I$ and $o \in O$ in the formula.
If the part of the formula $\varphi$ that relates multiple paths consists only of such constraints $\dep{J}{\set{o}}$ with the rest being an $\ltl$ property, we can interpret $\varphi$ as a description of a distributed architecture.
If furthermore, the $\dep{J_i}{\set{o_i}}$ constraints can be ordered such that $J_i\< J_{i+1}$ for all $i$, the architecture is linear.
There are three steps to check whether $\varphi$ is in the linear fragment:
\begin{enumerate}
	\item First, we have to add input-determinism to the formula $\varphi_\mathit{det} \coloneqq \varphi \land \dep{I}{O}$.
	      This preserves realizability as strategies are input-deterministic.
	\item Find for each output variable $o_i \in O$ possible sets of variables $J_i$, $o_i$ depends on, such that $J_i\< J_{i+1}$.
	      To check whether the choice of $J$'s is correct, we test if $\collapse(\varphi) \land \bigwedge_{o_i \in O} \dep{J_i}{\set{o_i}}$ is equivalent to $\varphi_\mathit{det}$.
	      This equivalence check is decidable as both formulas are in the universal fragment~\cite{conf/concur/FinkbeinerH16}.
	\item Finally, we construct the corresponding distributed realizability problem with linear architecture.
	      Formally, we define the distributed architecture $\arch = \tuple{P, p_{env},\mathcal{I},\mathcal{O}}$ with $P = \set{p_i \mid o_i \in O} \cup \set{\penv}$, $\mathcal{I}(p_i) = J_i$, $\mathcal{O}(p_i) = \set{o_i}$, and $\mathcal{O}(\penv) = I$. The $\ltl$ specification for the distributed synthesis problem is $\collapse(\varphi)$
\end{enumerate}

\begin{definition}[linear fragment of $\forall^*$]
	A formula $\varphi$ is in the linear fragment of $\forall^*$ iff for all $o_i \in O$ there is a $J_i \< I$ such that $\varphi \land \dep{I}{O} \equiv \collapse(\varphi) \land \bigwedge_{o_i \in O} \dep{J_i}{\set{o_i}}$ and $J_i \< J_{i+1}$ for all $i$.
\end{definition}
%
%
Note, that each $\forall^1$ formula $\varphi$ (or $\varphi$ is collapsible to a $\forall^1$ formula) is in the linear fragment because we can set all $J_i=I$ and additionally $\collapse(\varphi)=\varphi$ holds.

As an example of a formula in the linear fragment of $\forall^*$, consider $\varphi = \forall\pi,\pi'\ldot \dep{\set{a}}{\set{c}} \land \G  (c_\pi\leftrightarrow d_\pi ) \land \G (b_\pi\leftrightarrow \X e_\pi) $ with $I=\set{a,b}$ and $O=\set{c,d,e}$.
The corresponding formula asserting input-deterministism is $\varphi_{det}=\varphi\land\dep{I}{O}$.
One possible choice of $J$'s is $\set{a,b}$ for $c$, $\set{a}$ for $d$ and $\set{a,b}$ for $e$.
Note, that one can use either $\set{a,b}$ or $\set{a}$ for $c$ as $\dep{\set{a}}{\set{d}} \land (c_\pi\leftrightarrow d_\pi )$ implies $\dep{\set{a}}{\set{c}}$.
However, the apparent alternative $\set{b}$ for $e$ would yield an undecidable architecture.
It holds that $\varphi_{det}$ and $\collapse(\varphi) \land \dep{\set{a,b}}{\set{c}} \land \dep{\set{a}}{\set{d}} \land \dep{\set{a,b}}{\set{e}}$ are equivalent and, thus, that $\varphi$ is in the linear fragment.


\begin{theorem}
	The linear fragment of universal $\hyperltl$ is decidable.
\end{theorem}

\begin{proof}
    It holds that $\varphi \equiv \collapse(\varphi) \land \bigwedge_{o_i \in O} \dep{J_i}{\set{o_i}}$ for some $J_i$'s.
    The $\ltl$ distributed realizability problem for $\collapse(\varphi)$ in the constructed architecture $A$ is equivalent to the $\hyperltl$ realizability of $\varphi$ as the architecture $\arch$ represents exactly the input-determinism represented by formula $\bigwedge_{o_i \in O} \dep{J_i}{\set{o_i}}$.
	The architecture is linear and, thus, the realizability problem is decidable.
\end{proof}

\subsubsection{$\exists^* \forall^1$ Fragment.}

In this fragment, we consider arbitrary many existential path quantifier followed by a single universal path quantifier.
This fragment turns out to be still decidable.
We solve the realizability problem for this fragment by reducing it to a decidable fragment of the distributed realizability problem.

\begin{theorem} \label{thm:decidability_ea}
	Realizability of $\exists^* \forall^1 \hyperltl$ specifications is decidable.
\end{theorem}
\begin{proof}
	Let $\varphi$ be $\exists \pi_1 \dots \exists \pi_n \forall \pi' \ldot \psi$.
	We reduce the realizability problem of $\varphi$ to the distributed realizability problem for $\ltl$.
	For every existential path quantifier $\pi_i$, we introduce a copy of the atomic propositions, written $a_{\pi_i}$ for $a \in \ap$.
	Intuitively, those select the paths in the strategy tree where the existential path quantifiers are evaluated.
	Thus, those propositions (1) have to encode an actual path in the strategy tree and (2) may not depend on the branching of the strategy tree.
	To ensure (1), we add the $\ltl$ constraint $\G (I_{\pi_i} = I_{\pi'}) \rightarrow \G (O_{\pi_i} = O_{\pi'})$ that asserts that if the inputs correspond to some path in the strategy tree, the outputs on those paths have to be the same.
	Property (2) is guaranteed by the distributed architecture, the processes generating the propositions $a_{\pi_i}$ do not depend on the environment output.
	The resulting architecture $\arch_\varphi$ is $\tuple{\set{\penv,p,p'},\penv,\set{ p \mapsto \emptyset, p' \mapsto I_{\pi'} }, \set{ \penv \mapsto I_{\pi'}, p \mapsto \bigcup_{1 \leq i \leq n} O_{\pi_i} \cup I_{\pi_i}, p' \mapsto O_{\pi'}}}$.
	It is easy to verify that $\arch_\varphi$ does not contain an information fork, thus the realizability problem is decidable.
	The $\ltl$ specification $\theta$ is $\psi \land \bigwedge_{1 \leq i \leq n} \G (I_{\pi_i} = I_{\pi'}) \rightarrow \G (O_{\pi_i} = O_{\pi'})$.
	The implementation of process $p'$ (if it exists) is a model for the $\hyperltl$ formula (process $p$ producing witness for the $\exists$ quantifier).
	Conversely, a model for $\varphi$ can be used as an implementation of $p'$.
	Thus, the distributed synthesis problem $\tuple{\arch_\varphi,\theta}$ has a solution if, and only if, $\varphi$ is realizable.
\end{proof}

\subsubsection{$\forall^* \exists^*$ Fragment.}

The last fragment to consider are formulas in the $\forall^*\exists^*$ fragment.
Whereas the $\exists^*\forall^1$ fragment remains decidable, the realizability problem of $\forall^*\exists^*$ turns out to be undecidable even when restricted to only one quantifier of both sorts ($\forall^1\exists^1$).

\begin{theorem}\label{pcp}
	Realizability of $\forall^* \exists^*$HyperLTL is undecidable.
\end{theorem}

\begin{proof}
	The proof is done via reduction from Post's Correspondence Problem (PCP)~\cite{post1946variant}. The basic idea follows the proof in~\cite{conf/concur/FinkbeinerH16}.
\end{proof}

%% file: tikz/undecidable_architecture.tex
\begin{tikzpicture}[->,>=stealth',shorten >=1pt,auto,semithick,scale=1,transform shape,scale=0.8]
	\node [state] (e) {$env$};
	\node [state, below left=1 of e] (a) {$p_1$};
	\node [state, below right=1 of e] (b) {$p_2$};
	\path[->]
	(e) edge node [label,above left = 0 and -0.1] {$a$} (a)
	(e) edge node [label,above right = 0 and -0.1] {$b$} (b)
	(a) edge node [label,above left = -0.15 and 0] {$c$} +(0,-1)
	(b) edge node [label,above right = -0.15 and 0] {$d$} +(0,-1)
	;
\end{tikzpicture}

%% file: tikz/incomplete_information.tex
\begin{tikzpicture}[->,>=stealth',shorten >=1pt,auto,semithick,scale=1,transform shape,scale=0.8]
	\node [state] (e) {$env$};
	\node [state, below left=1 of e] (a) {$p_1$};
	\node [state, below right=1 of e] (b) {$p_2$};
	\path[->]
	(e) edge node [label,above left = 0 and -0.1] {$a$} (a)
	(e) edge node [label,above right = 0 and -0.1] {$a,b$} (b)
	(a) edge node [label,above left = -0.15 and 0] {$c$} +(0,-1)
	(b) edge node [label,above right = -0.15 and 0] {$d$} +(0,-1)
	;
\end{tikzpicture}

%% file: bounded.tex
\section{Bounded Realizability} \label{sec:bounded-realizability}

In this section, we propose an algorithm to synthesize strategies from specifications given in universal $\hyperltl$ by searching for finite generators of realizing strategies.
We encode this search as a satisfiability problem for a decidable constraint system.

\vspace{-5pt}\paragraph{Transition Systems.}
A \emph{transition system} $\tsys$ is a tuple $\tuple{S,s_0,\tau,l}$ where
$S$ is a finite set of states,
$s_0 \in S$ is the designated initial state,
$\tau \colon \fun{S \times \pow{I}}{S}$ is the transition function, and
$l \colon \fun{S}{\pow{O}}$ is the state-labeling or output function.
We generalize the transition function to sequences over $\pow{I}$ by defining $\tau^* \colon \fun{(\pow{I})^*}{S}$ recursively as $\tau^*(\epsilon) = s_0$ and $\tau^*(w_0 \cdots w_{n-1} w_n) = \tau(\tau^*(w_0 \cdots w_{n-1}), w_n)$ for $w_0 \cdots w_{n-1} w_n \in (\pow{I})^+$.
A transition system $\tsys$ \emph{generates} the strategy $f$ if $f(w) = l(\tau^*(w))$ for every $w \in (\pow{I})^*$.
A strategy $f$ is called \emph{finite-state} if there exists a transition system that generates~$f$.


\vspace{-5pt}\paragraph{Overview.}

We first sketch the synthesis procedure and then proceed with a description of the intermediate steps.
Let $\varphi$ be a universal $\hyperltl$ formula $\forall \pi_1 \cdots \forall \pi_n \ldot \psi$.
We build the automaton $\mathcal{A}_\psi$ whose language is the set of tuples of traces that satisfy $\psi$.
We then define the acceptance of a transition system $\tsys$ on $\mathcal{A}_\psi$ by means of the self-composition of $\tsys$.
Lastly, we encode the existence of a transition system accepted by $\mathcal{A}_\psi$ as an SMT constraint system.

\begin{example} \label{ex:secret-decision}
  Throughout this section, we will use the following (simplified) running example.
  Assume we want to synthesize a system that keeps decisions secret until it is allowed to publish.
  Thus, our system has three input signals \emph{decision}, indicating whether a decision was made, the secret \emph{value}, and a signal to \emph{publish} results.
  Furthermore, our system has two outputs, a \emph{high} output \emph{internal} that stores the value of the last decision, and a \emph{low} output \emph{result} that indicates the result.
  No information about decisions should be inferred until publication.
  To specify the functionality, we propose the $\ltl$ specification
  \begin{align} \label{eq:example-ltl}
    &\G (\textit{decision} \rightarrow (value \leftrightarrow \X \textit{internal})) \nonumber \\
    \land &\G (\neg\textit{decision} \rightarrow (internal \leftrightarrow \X \textit{internal})) \nonumber \\
    \land &\G (\textit{publish} \rightarrow \X (\textit{internal} \leftrightarrow \textit{result})) \enspace.
  \end{align}
  The solution produced by the LTL synthesis tool BoSy~\cite{conf/cav/FaymonvilleFT17}, shown in Fig.~\ref{fig:bosy-secret-solution}, clearly violates our intention that results should be secret until publish: Whenever a decision is made, the result output is changed as well.
  
  We formalize the property that no information about the decision can be inferred from \emph{result} until publication as the $\hyperltl$ formula
  \begin{equation} \label{eq:example-hyperltl}
    \forall \pi \forall \pi' \ldot (\textit{publish}_\pi \lor \textit{publish}_{\pi'}) \R (\textit{result}_\pi \leftrightarrow \textit{result}_{\pi'}) \enspace.
  \end{equation}
  It asserts that for every pair of traces, the \emph{result} signals have to be the same until (if ever) there is a \emph{publish} signal on either trace.
  A solution satisfying both, the functional specification and the hyperproperty, is shown in Fig.~\ref{fig:bosy-secret-solution}.
  The system switches states whenever there is a decision with a different value than before and only exposes the decision in case there is a prior publish command.
\end{example}
\begin{figure}[t]
	\begin{subfigure}[b]{0.3\textwidth}
		\centering
		\input{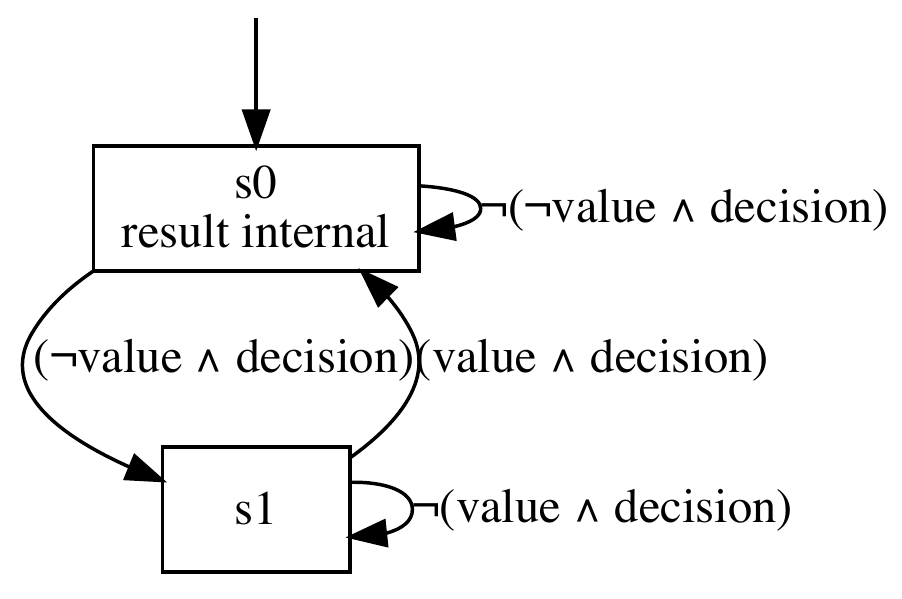}
		\caption{}
		\label{fig:bosy-secret}
	\end{subfigure}\hfill
	\begin{subfigure}[b]{0.69\textwidth}
		\centering
		\input{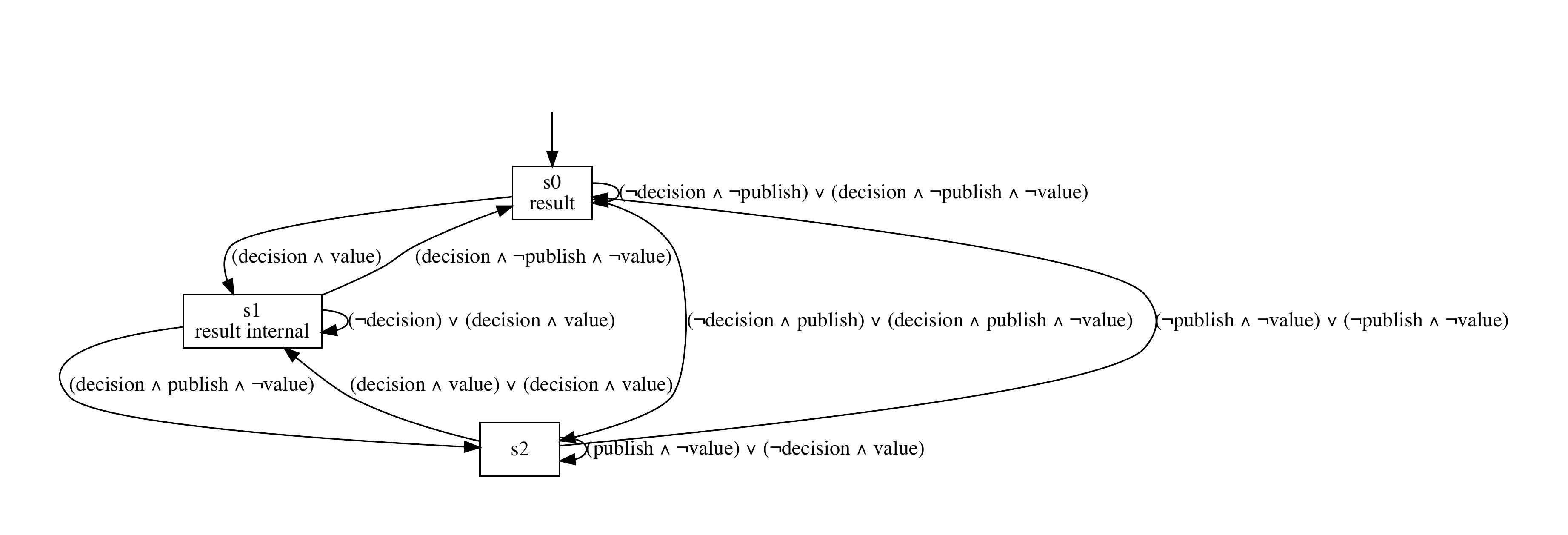}
		\caption{}
		\label{fig:bosyhyper-secret}
	\end{subfigure}
  \caption{Synthesized solutions for Example~\ref{ex:secret-decision}.}
  \label{fig:bosy-secret-solution}
  \vspace{-10pt}
\end{figure}
We proceed with introducing the necessary preliminaries for our algorithm.

\vspace{-5pt}\paragraph{Automata.}

A universal co-B\"uchi automaton $\ucw$ over a finite alphabet $\Sigma$ is a tuple $\tuple{Q,q_0,\delta,F}$, where
$Q$ is a finite set of states,
$q_0 \in Q$ is the designated initial state,
$\delta:Q \times 2^\Sigma \times Q$ is the transition relation, and
$F \subseteq Q$ is the set of rejecting states.
Given an infinite word $\sigma = \sigma_0 \sigma_1 \sigma_2 \cdots \in (2^\Sigma)^\omega$, a run of $\sigma$ on $\ucw$ is an infinite path $q_0 q_1 q_2 \dots \in Q^\omega$ where for all $i \geq 0$  it holds that $(q_i,\sigma_i,q_{i+1}) \in \delta$.
A run is accepting, if it contains only finitely many rejecting states.
$\ucw$ accepts a word $\sigma$, if \emph{all} runs of $\sigma$ on $\ucw$ are accepting.
The language of $\ucw$, written $\lang(\ucw)$, is the set $\set{\sigma \in (2^\Sigma)^\omega \mid \ucw \text{ accepts } \sigma}$.

We represent automata as directed graphs with vertex set $Q$ and a symbolic representation of the transition relation $\delta$ as propositional boolean formulas $\bool(\Sigma)$. The rejecting states in $F$ are marked by double lines.
The automata for the $\ltl$ and $\hyperltl$ specifications from Example~\ref{ex:secret-decision} are depicted in Fig.~\ref{fig:example-automata}.

\begin{figure}[t]
	\begin{subfigure}[b]{0.49\textwidth}
		\centering
		\input{tikz/example-automaton-linear}
		\caption{Automaton accepting language defined by $\ltl$ formula in (\ref{eq:example-ltl})}
		\label{fig:example-automaton-linear}
	\end{subfigure}\hfill
	\begin{subfigure}[b]{0.49\textwidth}
		\centering
		\input{tikz/example-automaton-hyper}
		\caption{Automaton accepting language defined by $\hyperltl$ formula in (\ref{eq:example-hyperltl})}
		\label{fig:example-automaton-hyper}
	\end{subfigure}
	\caption{Universal co-B\"uchi automata recognizing the languages from Example~\ref{ex:secret-decision}.}
	\label{fig:example-automata}
	\vspace{-10pt}
\end{figure}
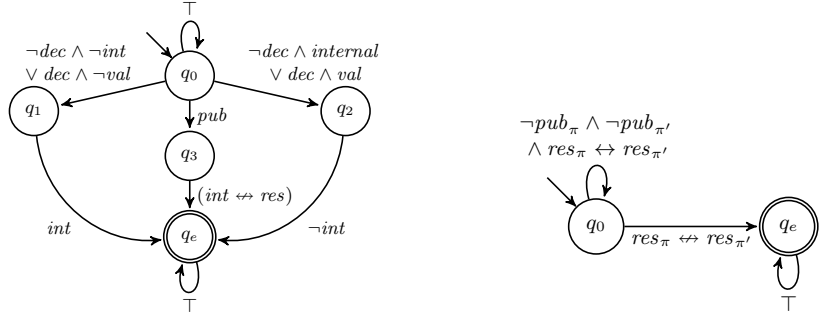

\vspace{-5pt}\paragraph{Run graph.}

The run graph of a transition system $\tsys = \tuple{S,s_0,\tau,l}$ on a universal co-B\"uchi automaton $\ucw = \tuple{Q,q_0,\delta,F}$ is a directed graph $\tuple{V,E}$ where
$V = S \times Q$ is the set of vertices and
$E \subseteq V \times V$ is the edge relation with
\begin{align*}
    ((s, q),(s',q')) \in E \;\text{ iff }\;\\
    \exists i \in 2^I \ldot \exists o \in 2^O \ldot (\tau(s,i) = s') \land (l(s) = o) \land (q,i \cup o,q') \in \delta \enspace .
\end{align*}
A run graph is accepting if every path (starting at the initial vertex $(s_0,q_0)$) has only finitely many visits of rejecting states.
To show acceptance, we annotate every reachable node in the run graph with a natural number $m$, such that any path, starting in the initial state, contains less than $m$ visits of rejecting states.
Such an annotation exists if, and only if, the run graph is accepting~\cite{journals/sttt/FinkbeinerS13}.

\vspace{-5pt}\paragraph{Self-composition.}

The model checking of universal $\hyperltl$ formulas~\cite{conf/cav/FinkbeinerRS15} is based on self-composition.
Let $\prj_i$ be the projection to the $i$-th element of a tuple.
Let $\zip$ denote the usual function that maps a $n$-tuple of sequences to a single sequence of $n$-tuples, for example, $zip([1, 2, 3], [4, 5, 6]) = [(1, 4), (2, 5), (3, 6)]$, and let $\unzip$ denote its inverse.
The transition system $\tsys^n$ is the $n$-fold self-composition of $\tsys = \tuple{S,s_0,\tau,l}$, if $\tsys^n = \tuple{S^n, s_0^n, \tau', l^n}$ and for all $s, s' \in S^n$, $\alpha \in (\pow{I})^n$, and $\beta \in (\pow{O})^n$ we have that $\tau'(s,\alpha) = s'$ and $l^n(s) = \beta$ iff for all $1 \leq i \leq n$, it hold that $\tau(\prj_i(s),\prj_i(\alpha)) = \prj_i(s')$ and $l(\prj_i(s)) = \prj_i(\beta)$.
If $T$ is the set of traces generated by $\tsys$, then $\set{ \zip(t_1,\dots,t_n) \mid t_1,\dots,t_n \in T }$ is the set of traces generated by $\tsys^n$.


We construct the universal co-B\"uchi automaton $\ucw_\psi$ such that the language of $\ucw_\psi$ is the set of words $w$ such that $\unzip(w) = \Pi$ and $\Pi \models_\emptyset \psi$, i.e., the tuple of traces that satisfy $\psi$.
We get this automaton by dualizing the non-deterministic B\"uchi automaton for $\neg\psi$~\cite{conf/post/ClarksonFKMRS14}, i.e., changing the branching from non-deterministic to universal and the acceptance condition from B\"uchi to co-B\"uchi.
Hence, $\tsys$ satisfies a universal $\hyperltl$ formula $\varphi = \forall \pi_1 \dots \forall \pi_k \ldot \psi$ if the traces generated by self-composition $\tsys^n$ are a subset of $\lang(\ucw_\psi)$.

\begin{lemma}
    A transition system $\tsys$ satisfies the universal $\hyperltl$ formula $\varphi = \forall \pi_1 \cdots \forall \pi_n \ldot \psi$, if the run graph of $\tsys^n$ and $\ucw_\psi$ is accepting.
\end{lemma}

\vspace{-5pt}\paragraph{Synthesis.}

Let $\tsys = \tuple{S,s_0,\tau,l}$ and $\ucw_\psi = \tuple{Q,q_0,\delta,F}$.
We encode the synthesis problem as an SMT constraint system.
Therefore, we use uninterpreted function symbols to encode the transition system and the annotation.
For the transition system, those functions are the transition function $\tau : S \times \pow{I} \rightarrow S$ and the labeling function $l: S \rightarrow \pow{O}$.
The annotation is split into two parts, a reachability constraint $\lambda^\bool : S^n \times Q \rightarrow \bool$ indicating whether a state in the run graph is reachable and a counter $\lambda^\# : S^n \times Q \rightarrow \nat$ that maps every reachable vertex to the maximal number of rejecting states visited by any path starting in the initial vertex.
The resulting constraint asserts that there is a transition system with accepting run graph.
\begin{align*}
  &\forall s, s' \in S^n \ldot
   \forall q, q' \in Q \ldot
   \forall i \in (\pow{I})^n \ldot\\
  &\left(
  \lambda^\bool(s, q) \land
  \tau'(s, i) = s' \land
  (q, i \cup l(s), q') \in \delta
  \right)
  \rightarrow \lambda^\bool(s', q') \land
  \lambda^\#(s', q') \trianglerighteq \lambda^\#(s, q)
\end{align*}
where $\trianglerighteq$ is $>$ if $q' \in F$ and $\geq$ otherwise.

\begin{theorem}
  The constraint system is satisfiable with bound $b$ if, and only if, there is a transition system $\tsys$ of size $b$ that realizes the $\hyperltl$ formula.
\end{theorem}
We extract a realizing implementation by asking the satisfiability solver to generate a model for the uninterpreted functions that encode the transition system. 

%% file: tikz/bosy-secret.tex
\begin{tikzpicture}[->,>=stealth',shorten >=1pt,auto,semithick,scale=1,transform shape,scale=0.8]
  \tikzstyle{state}=[draw,shape=rectangle,inner sep=7pt,minimum width=1cm,align=center]

  \node[state] (s0) {$s_0$\\ $\set{\textit{res}, \textit{int}}$};
  \node[state,right=1.5 of s0] (s1) {$s_1$\\ $\emptyset$};
  
  \draw (s0) edge[<-] +(0,1)
        (s0) edge[loop below] node {$\neg\textit{dec} \lor \textit{val}$} ()
        (s0) edge[bend left] node[anchor=south] {$\textit{dec} \land \neg\textit{val}$} (s1)
        (s1) edge[bend left] node[anchor=north] {$\textit{dec} \land \textit{val}$} (s0)
        (s1) edge[loop below] node {$\neg\textit{dec} \lor \neg\textit{val}$} ()
        ;
  
\end{tikzpicture}

%% file: tikz/bosyhyper-secret.tex
\begin{tikzpicture}[->,>=stealth',shorten >=1pt,auto,semithick,scale=1,transform shape,scale=0.8]
  \tikzstyle{state}=[draw,shape=rectangle,inner sep=3pt,minimum width=1cm,align=center]

  \node[state] (s0) {$s_0$\\ $\set{\textit{res}}$};
  \node[state,below left=2 and 1.5 of s0] (s1) {$s_1$\\ $\set{\textit{res},\textit{int}}$};
  \node[state,below right=2 and 1.5 of s0] (s2) {$s_2$\\ $\emptyset$};
  
  \draw (s0) edge[<-] +(0,1)
        (s0) edge[loop right] node {$\neg\textit{pub} \land (\neg\textit{dec} \lor \neg\textit{val})$} ()
        (s0) edge[bend right=10] node[near end,yshift=5pt] {$\textit{dec} \land \textit{val}$} (s1)
        (s1) edge[bend left=30] node[] {$\textit{dec} \land \neg\textit{val} \land \neg\textit{pub}$} (s0)
        (s1) edge[loop left] node[near start,anchor=north,yshift=-5pt] {$\neg\textit{dec} \lor \textit{val}$} ()
        (s0) edge[bend left=30] node {$\textit{pub} \land (\neg\textit{dec} \lor \neg\textit{val})$} (s2)
        (s2) edge[bend right=10] node[near end,xshift=3pt] {$\neg\textit{pub} \land \neg\textit{val}$} (s0)
        (s2) edge[loop right] node[near end,anchor=north west,align=center,xshift=-15pt] {$\textit{pub} \land \neg\textit{val}$\\${}\lor \neg\textit{dec} \land \textit{val}$} ()
        (s1) edge[bend left=10] node[swap] {$\textit{dec} \land \textit{pub} \land \neg\textit{val}$} (s2)
        (s2) edge[bend left=10] node[] {$\textit{dec} \land \textit{val}$} (s1)
        ;
  
\end{tikzpicture}

%% file: tikz/example-automaton-linear.tex
\begin{tikzpicture}[->,>=stealth',shorten >=1pt,auto,semithick,scale=1,transform shape,scale=0.8]
  \node[state] (init) {$q_0$};
  \node[state,below left=0 and 2 of init] (left) {$q_1$};
  \node[state,below right=0 and 2 of init] (right) {$q_2$};
  \node[state,below=0.5 of init] (publish) {$q_3$};
  \node[state,accepting,below=0.5 of publish] (error) {$q_e$};
  
  \draw (init) edge[<-] +(-0.75,0.75)
        (init) edge[loop above] node {$\btrue$} ()
        (init) edge node[swap,align=center,near start] {$\neg\mathit{dec} \land \neg \mathit{int}$\\${} \lor \mathit{dec} \land \neg \mathit{val}$} (left)
        (left) edge[bend right=45] node[swap] {$\mathit{int}$} (error)
        (init) edge node[align=center,near start] {$\neg\mathit{dec} \land  \mathit{internal}$\\${} \lor \mathit{dec} \land \mathit{val}$} (right)
        (right) edge[bend left=45] node {$\neg\mathit{int}$} (error)
        (init) edge node {$\mathit{pub}$} (publish)
        (publish) edge node {$(\mathit{int} \nleftrightarrow \mathit{res})$} (error)
        (error) edge[loop below] node {$\btrue$} ()
        ;
\end{tikzpicture}

%% file: tikz/example-automaton-hyper.tex
\begin{tikzpicture}[->,>=stealth',shorten >=1pt,auto,semithick,scale=1,transform shape,scale=0.9]
  \node[state] (init) {$q_0$};
  \node[state,accepting,right=2 of init] (error) {$q_e$};
  
  \draw (init) edge[<-] +(-0.75,0.75)
        (init) edge[loop above,align=center] node {$\neg \mathit{pub}_\pi \land \neg \mathit{pub}_{\pi'}$\\${}\land \mathit{res}_\pi \leftrightarrow \mathit{res}_{\pi'}$} ()
        (init) edge node[swap] {$\mathit{res}_\pi \nleftrightarrow \mathit{res}_{\pi'}$} (error)
        (error) edge[loop below] node {$\btrue$} ()
        ;
\end{tikzpicture}

%% file: unrealizability.tex
\section{Bounded Unrealizability} \label{sec:bounded-unrealizability}

So far, we focused on the positive case, providing an algorithm for finding small solutions, if they exist.
In this section, we shift to the case of detecting if a universal $\hyperltl$ formula is unrealizable.
We adapt the definition of counterexamples to realizability for $\ltl$~\cite{conf/tacas/FinkbeinerT14} to $\hyperltl$ in the following.
Let $\varphi$ be a universal $\hyperltl$ formula $\forall \pi_1 \cdots \forall \pi_n \ldot \psi$ over inputs $I$ and outputs $O$, a \emph{counterexample to realizability} is a set of input traces $\cexpaths \subseteq (2^I)^\omega$ such that for every strategy $f: \strat{I}{O}$ the labeled traces $\cexpaths^f \subseteq (2^{I \cup O})^\omega$ satisfy $\neg \varphi = \exists \pi_1 \cdots \exists \pi_n \ldot \neg\psi$.

\begin{proposition}
  A universal $\hyperltl$ formula $\varphi = \forall \pi_1 \cdots \forall \pi_n \ldot \psi$ is unrealizable if there is a counterexample $\cexpaths$ to realizability.
\end{proposition}
\begin{proof}
  For contradiction, we assume $\varphi$ is realizable by a strategy $f$.
  As $\cexpaths$ is a counterexample to realizability, we know $\cexpaths^f \models \exists \pi_1 \cdots \exists \pi_n \ldot \neg\psi$.
  This means that there exists an assignment $\pathassign_\cexpaths \in \pathvars \to \cexpaths^f$ with $\pathassign_\cexpaths \models_{\cexpaths^f} \neg \psi$.
  Equivalently $\pathassign_\cexpaths \nmodels_{\cexpaths^f} \psi$.
  Therefore, not all assignments $\pathassign \in \pathvars \to \cexpaths^f$ satisfy $\pathassign \models_{\cexpaths^f} \psi$.
  Which implies $\cexpaths^f \nmodels \forall \pi_1 \cdots \forall \pi_n \ldot \psi = \varphi$.
  Since $\varphi$ is universal, we can defer $f \nmodels \varphi$, which concludes the contradiction.
  Thus, $\varphi$ is unrealizable.
\end{proof}

Despite being independent of strategy trees, there are in many cases finite representations of $\cexpaths$.
Consider, for example, the unrealizable specification $\varphi_1 = \forall \pi \forall \pi' \ldot \F (i_\pi \leftrightarrow i_{\pi'})$, where the set $\cexpaths_1 = \set{ \emptyset^\omega, \set{i}^\omega }$ is a counterexample to realizability.
As a second example, consider $\varphi_2 = \forall \pi \forall \pi' \ldot \G (o_\pi \leftrightarrow o_{\pi'}) \land \G (i_\pi \leftrightarrow \X o_\pi)$ with conflicting requirements on $o$.
$\cexpaths_1$ is a counterexample to realizability for $\varphi_2$ as well: By choosing a different valuation of $i$ in the first step, the system is forced to either react with different valuations of $o$ (violating first conjunct), or not correctly repeating the initial value of $i$ (violating second conjunct).

There are, however, already linear specifications where the set of counterexample paths is not finite and depends on the strategy tree~\cite{journals/corr/FinkbeinerT15}.
For example, the specification $\forall \pi \ldot \F (i_\pi \leftrightarrow o_\pi)$ is unrealizable as the system cannot predict future values of the environment.
There is no finite set of traces witnessing this: For every finite set of traces, there is a strategy tree such that $\F (i_\pi \leftrightarrow o_\pi)$ holds on every such trace.
On the other hand, there is a simple \emph{counterexample strategy}, that is a strategy that observes output sequences and produces inputs. 
In this example, the counterexample strategy inverts the outputs given by the system, thus it is guaranteed that $\G (i \nleftrightarrow o)$ for any system strategy.


We combine those two approaches, selecting counterexample paths and using strategic behavior.
A $k$-counterexample strategy for $\hyperltl$ observes $k$ output sequences and produces $k$ inputs, where $k$ is a new parameter ($k \geq n$).
The counterexample strategy is winning if (1) either the traces given by the system player do not correspond to a strategy, or (2) the body of the $\hyperltl$ is violated for any $n$ subset of the $k$ traces.
Regarding property (1), consider the two traces where the system player produces different outputs initially.
Clearly, those two traces cannot be generated by any system strategy since the initial state (root labeling) is fixed.

The search for a $k$-counterexample strategy can be reduced to LTL synthesis using $k$-tuple input propositions $O^k$, $k$-tuple output propositions $I^k$, and the specification
\begin{equation*}
  \neg\dep{I^k}{O^k} \lor \bigvee_{P \subseteq \set{1,\dots,k} \text{ with } \card{P} = n} \neg\psi[P] \enspace,
\end{equation*}
where $\psi[P]$ denotes the replacement of $a_{\pi_i}$ by the $P_i$th position of the combined input/output $k$-tuple.

\begin{theorem}
  A universal $\hyperltl$ formula $\varphi = \forall \pi_1 \cdots \forall \pi_n \ldot \psi$ is unrealizable if there is a $k$-counterexample strategy for some $k \geq n$.
\end{theorem}

%% file: experiments.tex
\section{Evaluation} \label{sec:evaluation}

We implemented a prototype synthesis tool, called $\tool$\footnote{$\tool$ is available at \url{https://www.react.uni-saarland.de/tools/bosy/}}, for universal $\hyperltl$ based on the bounded synthesis algorithm described in Sec.~\ref{sec:bounded-realizability}.
Furthermore, we implemented the search for counterexamples proposed in Sec.~\ref{sec:bounded-unrealizability}.
Thus, $\tool$ is able to characterize realizability and unrealizability of universal $\hyperltl$ formulas.

We base our implementation on the LTL synthesis tool BoSy~\cite{conf/cav/FaymonvilleFT17}.
For efficiency, we split the specifications into two parts, a part containing the linear (LTL) specification, and a part containing the hyperproperty given as $\hyperltl$ formula.
Consequently, we build two constraint systems, one using the standard bounded synthesis approach~\cite{journals/sttt/FinkbeinerS13} and one using the approach described in Sec.~\ref{sec:bounded-realizability}.
Before solving, those constraints are combined into a single SMT query.
This results in a much more concise constraint system compared to the one where the complete specification is interpreted as a $\hyperltl$ formula.
For solving the SMT queries, we use the Z3 solver~\cite{conf/tacas/MouraB08}.
We continue by describing the benchmarks used in our experiments.

\vspace{-5pt}\paragraph{Symmetric mutual exclusion.}

Our first example demonstrates the ability to specify symmetry in $\hyperltl$ for a simple mutual exclusion protocol.
Let $r_1$ and $r_2$ be input signals representing mutual exclusive \emph{requests} to a critical section and $g_1$/$g_2$ the respective grant to enter the section.
Every request should be answered eventually $\G (r_i \rightarrow \F g_i)$ for $i \in \set{1,2}$, but not at the same time $\G \neg(g_1 \land g_2)$.
The minimal LTL solution is depicted in Fig.~\ref{fig:mutex-ltl}.
It is well known that no mutex protocol can ensure perfect symmetry~\cite{books/daglib/0080029}, thus when adding the symmetry constraint specified by the $\hyperltl$ formula $\forall \pi \forall \pi' \ldot ({r_1}_\pi \nleftrightarrow {r_2}_{\pi'}) \R ({g_1}_\pi \leftrightarrow {g_2}_{\pi'})$ the formula becomes unrealizable.
Our tool produces the counterexample shown in Fig.~\ref{fig:mutex-symmetric-counterexample}.
By adding another input signal \emph{tie} that breaks the symmetry in case of simultaneous requests and modifying the symmetry constraint $\forall \pi \forall \pi' \ldot \left( ({r_1}_\pi \nleftrightarrow {r_2}_{\pi'}) \lor (\mathit{tie}_\pi \nleftrightarrow \neg \mathit{tie}_{\pi'}) \right) \R ({g_1}_\pi \leftrightarrow {g_2}_{\pi'})$ we obtain the solution depicted in Fig.~\ref{fig:mutex-symmetry-tie}.
We further evaluated the same properties on a version that forbids spurious grants, which are reported in Table~\ref{tbl:results} with prefix \emph{full}.

\begin{figure}[t]
  \begin{subfigure}[b]{0.3\textwidth}
    \centering
    \input{tikz/mutex-ltl-solution}
    \caption{Non-symmetric solution}
    \label{fig:mutex-ltl}
  \end{subfigure}
  \begin{subfigure}[b]{0.3\textwidth}
    \centering
    \input{tikz/mutex-sym-unrealizable}
    \caption{Counterexample to symmetry}
    \label{fig:mutex-symmetric-counterexample}
  \end{subfigure}
  \begin{subfigure}[b]{0.3\textwidth}
    \centering
    \input{tikz/mutex-symmetry-tie}
    \caption{Symmetry breaking solution}
    \label{fig:mutex-symmetry-tie}
  \end{subfigure}
  \caption{Synthesized solution of the mutual exclusion protocols.}
  \vspace{-10pt}
\end{figure}
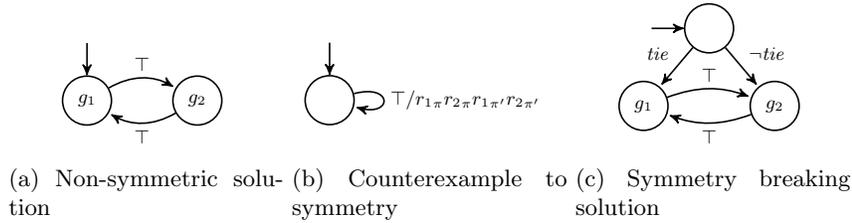

\vspace{-5pt}\paragraph{Distributed and fault-tolerant systems.}

In Sec.~\ref{sec:decidability} we have shown how to encode arbitrary distributed architectures in $\hyperltl$.
As an example for our evaluation, consider a setting with two processes, one for \emph{encoding} input signals and one for \emph{decoding} the encoded signals.
Both processes can be synthesized simultaneously using a single $\hyperltl$ specification.
The (linear) correctness condition states that the decoded signal is always equal to the inputs given to the encoder.
Furthermore, the encoder and decoder should solely depend on the inputs and the encoded signal, respectively.
Additionally, we can specify desired properties about the encoding like fault-tolerance~\cite{journals/corr/FinkbeinerT15} or Hamming distance of code words~\cite{conf/cav/FinkbeinerRS15}.
The results are reported in Table~\ref{tbl:results} where $i$-$j$-$x$ means $i$ input bits, $j$ encoded bits, and $x$ represents the property.
The property is either tolerance against a single Byzantine signal failure or a guaranteed Hamming distance of code words.

%

\vspace{-5pt}\paragraph{CAP Theorem.}

The CAP Theorem due to Brewer~\cite{conf/podc/Brewer00} states that it is impossible to design a distributed system that provides Consistency, Availability, and Partition tolerance (CAP) simultaneously.
This example has been considered before~\cite{journals/corr/FinkbeinerT15} to evaluate a technique that could automatically detect unrealizability.
However, when we drop either Consistency, Availability, or Partition tolerance, the corresponding instances (AP, CP, and CA) become realizable, which the previous work was not able to prove.
We show that our implementation can show both, unrealizability of CAP and realizability of AP, CP, and CA.
In contrast to the previous encoding~\cite{journals/corr/FinkbeinerT15} we are not limited to acyclic architectures.

\vspace{-5pt}\paragraph{Long-term information-flow.}

Previous work on model-checking hyperproperties~\cite{conf/cav/FinkbeinerRS15} found that an implementation for the commonly used \emph{I2C} bus protocol could remember input values ad infinitum.
For example, it could not be verified that information given to the implementation eventually leaves it, i.e., is forgotten.
This is especially unfortunate in high security contexts.
We consider a simple bus protocol which is inspired by the widely used \emph{I2C} protocol.
%
Our example protocol has the inputs \emph{send} for initiating a transmission, \emph{in} for the value that should be transferred, and an \emph{ack}nowledgment bit indicating successful transmission.
The bus master waits in an \emph{idle} state until a \emph{send} is received.
Afterwards, it transmits a header sequence, followed by the value of \emph{in}, waits for an acknowledgement and then indicates \emph{success} or \emph{failure} to the sender before returning to the idle state.
We specify the property that the \emph{in}put has no influence on the \emph{data} that is send, which is obviously violated (instance NI1).
As a second property, we check that this information leak cannot happen arbitrary long (NI2) for which there is a realizing implementation.

\vspace{-5pt}\paragraph{Dining Cryptographers.}
Recap the dining cryptographers problem introduced earlier.
This benchmark is interesting as it contains two types of hyperproperties.
First, there is information-flow between the three cryptographers, where some secrets ($s_{ab}, s_{ac}, s_{bc}$) are shared between pairs of cryptographers.
In the formalization, we have 4 entities: three processes describing the 3 cryptographers ($\textit{out}_i$) and one process computing the result ($p_g$), i.e., whether the group has paid or not, from $\textit{out}_i$.
Second, the final result should only disclose whether one of the cryptographers has paid or the NSA.
This can be formalized as a indistinguishability property between different executions.
For example, when we compare the two traces $\pi$ and $\pi'$ where $C_a$ has paid on $\pi$ and $C_b$ has paid on $\pi'$.
Then the outputs of both have to be the same, if their common secret $s_{ab}$ is different on those two traces (while all other secrets $s_{ac}$ and $s_{bc}$ are the same).
This ensures that from an outside observer, a flipped output can be either result of a different shared secret or due to the announcement.
Lastly, the linear specification asserts that $p_g \leftrightarrow \neg p_\textit{NSA}$.

\vspace{-5pt}\paragraph{Results.}

Table~\ref{tbl:results} reports on the results of the benchmarks.
We distinguish between state-labeled (\emph{Moore}) and transition-labeled (\emph{Mealy}) transition systems.
Note that the counterexample strategies use the opposite transition system, i.e., a Mealy system strategy corresponds to a state-labeled (Moore) environment strategy.
Typically, Mealy strategies are more compact, i.e., need smaller transition systems and this is confirmed by our experiments.
BoSyHyper is able to solve most of the examples, providing realizing implementations or counterexamples.
Regrading the unrealizable benchmarks we observe that usually two simultaneously generated paths ($k=2$) are enough with the exception of the encoder example.
Overall the results are encouraging showing that we can solve a variety of instances with non-trivial information-flow.

\begin{table}[t]
  \caption{Results of BoSyHyper on the benchmarks sets described in Sec.~\ref{sec:evaluation}. They ran on a machine with a dual-core Core i7, 3.3\,GHz, and 16\,GB memory. }
  \label{tbl:results}
  \centering
  \scalebox{0.9}{
  \begin{tabular}{lllrrrr}
    \hline\noalign{\smallskip}
    Benchmark & Instance & Result & \multicolumn{2}{c}{States} & \multicolumn{2}{c}{Time[sec.]} \\
     &&& Moore & Mealy &\quad Moore& Mealy \\
    \noalign{\smallskip}\hline\hline\noalign{\smallskip}
    \multirow{6}{*}{Symmetric Mutex}
     & non-sym      & realizable           & 2 & 2 & 1.4 & 1.3 \\
     & sym          & unrealizable ($k=2$) & 1 & 1 & 1.9 & 2.0 \\
     & tie          & realizable           & 3 & 3 & 1.7 & 1.6 \\
     & full-non-sym & realizable           & 4 & 4 & 1.4 & 1.4 \\
     & full-sym     & unrealizable ($k=2$) & 1 & 1 & 4.3 & 6.2 \\
     & full-tie     & realizable           & 9 & 5 & 1\,802.7 & 5.2 \\
    \noalign{\smallskip}\hline\noalign{\smallskip}
    \multirow{6}{*}{Encoder/Decoder}
     & 1-2-hamming-2      & realizable           & 4  & 1 & 1.6      & 1.3 \\
     & 1-2-fault-tolerant & unrealizable ($k=2$) & 1  & - & 54.9     & - \\
     & 1-3-fault-tolerant & realizable           & 4  & 1 & 151.7    & 1.7 \\
     & 2-2-hamming-2      & unrealizable $(k=3)$ & -  & 1 & -        & 10.6 \\
     & 2-3-hamming-2      & realizable           & 16 & 1 & $> 1\,h$ & 1.5 \\
     & 2-3-hamming-3      & unrealizable $(k=3)$ & -  & 1 & -        & 126.7 \\
    \noalign{\smallskip}\hline\noalign{\smallskip}
    \multirow{8}{*}{CAP Theorem}
     & cap-2-linear & realizable & 8 & 1 & 7.0 & 1.3 \\
     & cap-2        & unrealizable $(k=2)$ & 1 & - & 1\,823.9 & - \\
     & ca-2         & realizable           & - & 1 & - & 4.4 \\
     & ca-3         & realizable           & - & 1 & - & 15.0 \\
     & cp-2         & realizable           & 1 & 1 & 1.8 & 1.6 \\
     & cp-3         & realizable           & 1 & 1 & 3.2 & 10.6 \\
     & ap-2         & realizable           & - & 1 & - & 2.0 \\
     & ap-3         & realizable           & - & 1 & - & 43.4 \\
    \noalign{\smallskip}\hline\noalign{\smallskip}
    \multirow{2}{*}{Bus Protocol}
     & NI1 & unrealizable $(k=2)$ & 1 & 1 & 75.2 & 69.6 \\
     & NI2 & realizable           & 8 & 8 & 24.1 & 33.9 \\
    \noalign{\smallskip}\hline\noalign{\smallskip}
    Dining Cryptographers & secrecy & realizable & - & 1 & - & 82.4 \\
    \hline\noalign{\smallskip}
  \end{tabular}}
  \vspace{-12pt}
\end{table}

%% file: tikz/mutex-ltl-solution.tex
\begin{tikzpicture}[->,>=stealth',shorten >=1pt,auto,node distance=1cm,semithick,scale=1,transform shape,scale=.8]

  \node[state] (init) {$g_1$};
  \node[state,right=of init] (other) {$g_2$};
  
  \draw (init) edge[<-] +(0,1)
        (init) edge[bend left] node {$\btrue$} (other)
        (other) edge[bend left] node {$\btrue$} (init)
        ;
  
\end{tikzpicture}

%% file: tikz/mutex-sym-unrealizable.tex
\begin{tikzpicture}[->,>=stealth',shorten >=1pt,auto,node distance=1cm,semithick,scale=1,transform shape,scale=.8]

  \node[state] (init) {};
  
  \draw (init) edge[<-] +(0,1)
        (init) edge[loop right] node {$\btrue/{r_1}_\pi {r_2}_\pi {r_1}_{\pi'} {r_2}_{\pi'}$} ()
        ;
  
  \node[below=17pt] {}; 
  
\end{tikzpicture}

%% file: tikz/mutex-symmetry-tie.tex
\begin{tikzpicture}[->,>=stealth',shorten >=1pt,auto,node distance=1cm,semithick,scale=1,transform shape,scale=.8]

  \node[state] (init) {};
  \node[state,below left=0.7 and 0.5 of init] (g1) {$g_1$};
  \node[state,below right=0.7 and 0.5 of init] (g2) {$g_2$};
  
  \draw (init) edge[<-] +(-1,0)
        (init) edge node[swap] {$\mathit{tie}$} (g1)
        (init) edge node {$\neg\mathit{tie}$} (g2)
        (g1) edge[bend left=20] node {$\btrue$} (g2)
        (g2) edge[bend left=20] node {$\btrue$} (g1)
        ;
  
\end{tikzpicture}

%% file: conclusion.tex
\section{Conclusion}

In this paper, we have considered the reactive realizability problem for specifications given in the temporal logic $\hyperltl$.
We gave a complete characterization of the decidable fragments based on the quantifier prefix and, additionally, identified a decidable fragment in the, in general undecidable, universal fragment of $\hyperltl$.
Furthermore, we presented two algorithms to detect realizable and unrealizable $\hyperltl$ specifications, one based on bounding the system implementation and one based on bounding the number of counterexample paths.
Our prototype implementation shows that our approach is able to synthesize systems with complex information-flow properties.